\documentclass[%
 reprint,
 amsmath,amssymb,
 aps,
]{revtex4-2}

\usepackage{tikz} 
\usepackage{amsmath,amsfonts,amsthm, mathtools}
\usepackage{float}
\usepackage{dsfont}
\usepackage{csquotes}
\usepackage{amssymb}
\usepackage{verbatim}
\usepackage{mathtools}
\usepackage{rotating}
\usepackage{bm}
\usepackage{hyperref}
\usepackage{mathtools}
\definecolor{darkred}  {rgb}{0.5,0,0}
\definecolor{darkblue} {rgb}{0,0,0.5}
\definecolor{darkgreen}{rgb}{0,0.5,0}
\hypersetup{
  colorlinks = true,
  urlcolor  = blue,         
  linkcolor = darkblue,     
  citecolor = darkgreen,    
  filecolor = darkred       
}
\theoremstyle{definition}
\newtheorem{observation}{Observation}

\newtheorem{conjecture}{Conjecture}
\newtheorem{lemma}{Lemma}
\newtheorem{proposition}{Proposition}
\newtheorem{theorem}{Theorem}

\newtheorem*{remark}{Remark}

\newcommand{\mbb}{\mathbb}
\newcommand{\mc}{\mathcal}

\newcommand{\mf}{\mathfrak}

\newcommand{\tr}{\textrm{Tr}}
\newcommand{\wt}{\widetilde}

\usepackage{multirow}
\newcommand{\ket}[1]{|#1\rangle}

\newcommand{\op}[2]{|#1\rangle\langle#2|}

\newcommand{\x}{\mathbf{x}}

\definecolor{cool_green}{rgb}{0.0, 0.5, 0.0}

\usepackage{blkarray}
\usepackage{adjustbox}
\usepackage{graphicx}
\usepackage{dcolumn}
\usepackage{bm}

\begin{document}

\preprint{APS/123-QED}

\title{Exact Steering Bound for Two-Qubit Werner States}

\author{Yujie Zhang}
\email{yujie4@illinois.edu}
 \affiliation{Department of Physics, University of Illinois at Urbana-Champaign, Urbana, IL 61801, USA}
\author{Eric Chitambar}
\email{echitamb@illinois.edu}
\affiliation{
Department of Electrical and Computer Engineering, University of Illinois at Urbana-Champaign, Urbana, IL 61801, USA
}%
\date{\today}

\begin{abstract}
We investigate the relationship between projective measurements and positive operator-valued measures (POVMs) in the task of quantum steering.  A longstanding open problem in the field has been whether POVMs are more powerful than projective measurements for the steerability of noisy singlet states, which are known as Werner states.  We resolve this problem for two-qubit systems and show that the two are equally powerful, thereby closing the so-called Werner gap.  Using the incompatible criteria for noisy POVMs and the connection between quantum steering and measurement incompatibility, we construct a local hidden state model for Werner states with Bloch sphere radius $r\leq 1/2$ under general POVMs.  This construction also provides a local hidden variable model for a larger range of Werner states than previously known. {In contrast, we also show that projective measurements and POVMs can have inequivalent noise tolerances when using a fixed state ensemble to build different local hidden state models.}  These results help clarify the relationship between projective measurements and POVMs for the tasks of quantum steering and nonlocal information processing.

\end{abstract}
\maketitle
Entanglement has always been one of the most puzzling features of quantum mechanics \cite{Horodecki2009,Einstein1935}.  Most notably, entanglement enables quantum behavior that is completely inconsistent with classical mechanics or any other theory satisfying the principle of locality \cite{Bell1964}. One such behavior, known as Bell nonlocality, has been extensively studied \cite{Brunner2014}, demonstrated experimentally \cite{Aspect1981, Clauser1969} and utilized in different quantum information applications \cite{Ekert1992, Masanes2011, Buhrman-1997a, Colbeck-2011a}.  While entanglement is a necessary and sufficient ingredient for demonstrating nonlocality in pure states \cite{Gisin1991, Gisin1992}, the relationship between entanglement and nonlocality in general mixed states is much less clear \cite{Popescu1995, Werner1989, Barrett2002, Methot-2007a}.  Most notable is the existence of Bell local entangled states, which are those that cannot generate Bell nonlocality by themselves.  The nature of nonlocality becomes even more interesting when considering a more general nonlocal effect known as quantum steering, which can be realized even for some Bell local states  \cite{Wiseman2007, Jones2007}.  As originally envisioned by Schrödinger in 1935 \cite{schrodinger1935, Wiseman2007}, {quantum steering involves a type of remote state preparation, and it is arguably the closest realization of Einstein's ``spooky action at a distance.''}

\par 
Understanding the differences between entanglement, steering, and Bell nonlocality has been a fundamental challenge in quantum information science.  The seminal work of Werner showed that not all entangled states are capable of generating Bell nonlocality by the explicit construction of local hidden variable (LHV) models \cite{Werner1989}.  The latter refers to theoretical models that reproduce the local measurement statistics of certain quantum states while still satisfying the principle of locality.  Werner's original model only considered local projection-valued measures (PVMs), but Barrett later extended it to account for the most general type of local measurements, which are those described by positive operator-valued measures (POVMs) \cite{Barrett2002}. Moreover, Werner and Barrett's models are even stronger in that they constitute what is now called a local hidden state (LHS) model. {Such models simulate not only the local measurement statistics but also the post-measurement quantum states for one party conditioned on the local measurement outcome of the other.  States satisfying an LHS model are called unsteerable.  }

Since the work of Werner and Barrett, significant advances have been made in the construction of both LHV and LHS models \cite{Hirsch2013, Augusiak2014, Quintino2015, Bowles2014, Hirsch2016, Cavalcanti2016, Nguyen2020, Nguyen2020a}.  {Some of these models hold only for PVMs, while others encompass POVMs as well.  The distinction between PVMs versus POVMs is crucial from a fundamental perspective.  While PVMs are experimentally much simpler to implement, a full accounting of what quantum mechanics allows under local processing should include the use of local ancilla, the enabling ingredient for POVMs.  Thus, the study of nonlocality is incomplete if it is just limited to PVMs.  

Focusing here on the case of steerability, this motivates the question of whether PVMs are strong enough on their own to separate the class of steerable states from unsteerable ones.  This question has remained unsolved for even for the simplest scenario of two-qubit Werner states, the canonical family of states for investigating nonlocality due to its analytical simplicity and deep connections to other aspects of quantum information theory.  So important is the steerability question for two-qubit Werner states that it currently sits on the Open Quantum Problems List, maintained by the Institute for Quantum Optics and Quantum Information (IQOQI) in Vienna (Problem 39 \cite{IQOQI2017}).}  Strong numerical evidence \cite{Nguyen2018} and explicit constructions for some special cases \cite{Werner2014} suggest that POVMs actually provide no advantage over PVMs in qubit steering of Werner states.  But a full solution to the problem has remained elusive, as well as a systematic way of constructing a LHS model for POVMs.  This paper resolves these open questions and proves that POVMs and PVMs are indeed equivalent for the steerability of two-qubit Werner states. 

\par Our method is based on the connection between quantum steering and measurement incompatibility \cite{Busch1986, Ali2009}, which has been established in the general setting 
\cite{Quintino2014, Uola2014} and further refined in the case of finite resources \cite{Bowles2015, zhang2023}.  We specifically focus on the fact that a POVM LHS model exists for a Werner state of radius $r$ if and only if there exists a compatible model for any set of noisy POVMs of the form $\{M^r_{a|x}:=rM_{a|x}+(1-r)\tfrac{\tr M_{a|x}}{2}\mbb{I}\}_{a,x},$ where $x$ denotes the POVM in the family, $a$ denotes the measurement outcome, $M_{a|x}\ge 0$, and $\sum_a M_{a|x}=\mbb{I}$ for all $x$.  In this letter, we will prove the compatibility of all noisy POVMs at radius $r=1/2$ by constructing an explicit compatibility model. The latter can then be immediately used as an LHS model for all qubit Werner states with $r\le 1/2$ subject to general measurements, thus implying their unsteerability using POVMs.  

\section{Results:} 
In quantum steering \cite{Uola2020} a bipartite state  $\rho_{AB}$ is shared between two observers, Alice and Bob. When Alice implements a local measurement chosen from some family $\{M_{a|x}\}_{a,x}$, the possible post-measurement states for Bob's system is given by the state assemblage $\{\sigma_{a|x}\}_{a,x}$, where
\begin{equation}
\label{eq:state ensemble}
\sigma_{a|x}=\tr_A[(M_{a|x}\otimes\mbb{I})\rho_{AB}].
\end{equation}
The state assemblage is unsteerable if it admits a local hidden state (LHS) model:
\begin{equation}
\sigma_{a|x}=\int d\lambda p(\lambda)p(a|x,\lambda)\rho_\lambda,
\label{Eq:LHS}
\end{equation}
where $p(\lambda)$ is a probability density function over variable $\lambda$ shared between Alice and Bob, $p(a|x,\lambda)$ is a (stochastic) response function for Alice, and $\rho_\lambda$ is a (hidden) state for Bob satisfying $\int d\lambda p(\lambda)\rho_\lambda=\rho_B$. 
 
Similarly, a family of POVMs $\{M_{a|x}\}_{a,x}$ is defined to be jointly measurable (compatible) if there exists a compatible model:
\begin{equation}
M_{a|x}=\int d\lambda p(a|x,\lambda)\Pi_\lambda
\label{Eq:compatible}
\end{equation}
with response functions $p(a|x,\lambda)$ and ``parent'' POVM $\{\Pi_\lambda\}_\lambda$.  The similar forms of Eqns. \eqref{Eq:LHS} and \eqref{Eq:compatible} is no coincidence, and a one-to-one correspondence between quantum steering and measurement incompatibility has been previously established \cite{Quintino2014, Uola2014, Heinosaari2015, zhang2023}.  Here we restate it in terms of the two-qubit Werner state $\rho_W(r)$ and an arbitrary noisy qubit POVM $\{M_a^r\}_a$, both of radius $r$ and respectively defined as
\begin{align}
        \rho_W(r)&=r\op{\Psi_-}{\Psi_-}+(1-r)\frac{\mbb{I}\otimes\mbb{I}}{4}\label{eq:Werner}\\
        M_{a}^r&=rM_{a}+(1-r)\frac{{\tr{(M_{a})}\mbb{I}}}{2},\label{Eq:noisy-POVM-defn}
\end{align}  
where $\ket{\Psi_-}_{AB}=\frac{1}{\sqrt{2}}(\ket{01}-\ket{10})$.
\begin{lemma}
A Werner state $\rho_W(r)$ can be simulated by an LHS model if and only if there exists a parent POVM that can simulate any noisy POVM $\{M_{a}^r\}_a$. {Moreover, the response functions $p(a|x,\lambda)$ in the two simulations are the same while the parent POVM $\{\Pi_\lambda\}_\lambda$ for the $\{M_{a}^r\}_a$ corresponding to a LHS ensemble for $\rho_W(r)$ with 
\begin{align}
\rho_\lambda&=r\frac{\sigma_y\Pi_\lambda^{\mathsf{T}}\sigma_y}{\tr[\Pi_\lambda]}+(1-r)\frac{\mbb{I}}{2},&p(\lambda)&=\frac{\tr[\Pi_\lambda]}{2},\notag
\end{align} 
and a similar correspondence holds in the other direction.}
\label{lem:bound}
\end{lemma}
\begin{figure}[t]
    \centering
\includegraphics[width=0.48\textwidth]{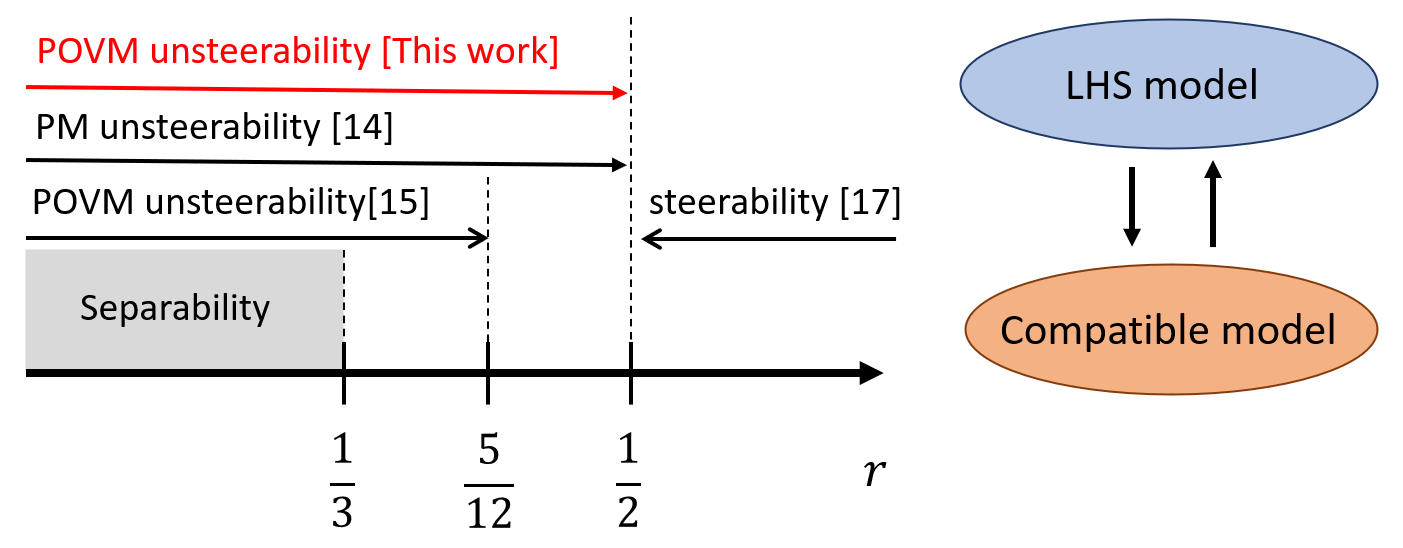}
    \caption{Left: Steering bounds for two-qubit Werner states $\rho_W(r)$ under projective measurements (PVMs) and POVMs. {This work fully closes the ``Werner gap'' by raising the POVM lower bound to $r=1/2$.}; Right: The presence of a LHS model for noisy Werner state under to PVMs/POVMs implies a compatible model for the set of all noisy PVMs/POVMs and vice versa.}
    \label{fig:steering bound}
\end{figure}
\begin{remark}
In Lemma ~\ref{lem:bound}, it is sufficient to just consider noisy versions of extremal POVMs.  For qubits, every extremal POVM consists of no more than four effects $M_a$ that are each rank-one \cite{Ariano2005}.  
We can therefore restrict our analysis to POVM effects expressible in the Pauli basis as $M_{a}=\mu_{a}(\mbb{I}+\hat{m}_{a}\cdot 
\vec{\sigma}$), where $\vec{\sigma}=[\sigma_x,\sigma_y,\sigma_z]$ is the standard vector of Pauli matrices and $|\hat{m}_a|=1$.  From positivity and normalization, one has $\sum_a\mu_a=1$, $\sum_a\mu_a\hat{m}_a=\vec{0}$ and $\mu_a\ge 0$, which together also imply that $\mu_{a}\le 1/2$.  {Note that PVMs are a special subset of extermal POVMs consisting of two effects $M_{\pm}=1/2(\mbb{I}\pm \hat{m}\cdot\vec{\sigma})$. When expressed in the Pauli basis, Eq. \eqref{Eq:noisy-POVM-defn} for extremal POVMs then becomes $M^r_a=\mu_{a}(\mbb{I}+r\hat{m}_{a}\cdot \vec{\sigma})$.
}
\end{remark}

Extensive research has been devoted to understanding the steering and Bell nonlocal bounds for Werner states \cite{Werner1989, Barrett2002, Augusiak2014, Hua2015, Hirsch2017, Divianszky2017, Nguyen2018a, designolle2023}. Compared to the Bell nonlocal threshold, the steering threshold for two-qubit Werner states is better understood, and the exact value of $r=1/2$ has been proven for PVMs \cite{Wiseman2007}. However, for POVMs the steering bound is still yet unknown \cite{IQOQI2017}.  Here we approach this problem by focusing on measurement compatibility and characterizing the bound when all noisy POVMs became compatible. From Lemma.~\ref{lem:bound}, this will lead to a steering bound for the Werner state. 
\par

\textbf{Compatible model --} The compatibility problem and compatible models for noisy qubit projective measurements $\{M^r_{\pm}=1/2(\mbb{I}\pm r\hat{m}\cdot\sigma)\}_{\hat{n}}$ have been developed and used to study different problems \cite{Uola2016, Bavaresco2017}, and it was shown in our previous work \cite{zhang2023} that there exists a compatible model for the whole family of noisy projective measurements using a finite-sized parent POVM whenever $r< 1/2$, while at $r=1/2$, the infinite-outcome parent POVM $\{\Pi_{\hat{n}}=\frac{1}{4\pi}(\mbb{I}+\hat{n}\cdot\vec{\sigma})\}_{\hat{n}}$ provides a simulation.  Explicitly, a compatible model that simulates a noisy PVM at $r=\frac{1}{2}$ in any spin direction $\hat{m}$ is given by
\begin{align}
M^{r=1/2}_{\pm}&=\int_{\mc{S}} d\hat{n} p(\pm|\hat{m},\hat{n}) \Pi_{\hat{n}}=\int_{\mc{S}} d\hat{n} \Theta(\pm\hat{m}\cdot\hat{n})\Pi_{\hat{n}}\notag\\
&= \int_{\mc{S}} d\hat{n}\Theta(\pm\hat{m}\cdot\hat{n})\frac{\mbb{I}+\hat{n}\cdot\vec{\sigma}}{4\pi}=\frac{1}{2}(\mbb{I}\pm\frac{1}{2}\hat{m}\cdot\vec{\sigma}),
\label{eq:full-compatible model}
\end{align}
where $p(\pm|\hat{m},\hat{n})=\Theta(\pm\hat{m}\cdot\hat{n})$ with $\Theta$ being the Heaviside step function, and the integration is taken over the surface of the entire Block sphere $\mc{S}$.  Our goal is to generalize this model for the simulation of noisy POVMs.  Unlike Barrett's model \cite{Barrett2002}, we will keep the same radius $r$ in this modification.  

As a starting point, we use the same parent POVM $\{\Pi_{\hat{n}}\}_{\hat{n}}$ and try to extend the PVM model to general POVMs $\{M^r_a=\mu_a(\mbb{I}+r\hat{m}_a\cdot\vec{\sigma})\}_a$ by using the the response function to be $p(a|\{M^r_a\}_a,\hat{n})=2\mu_a\Theta(\hat{m}_a\cdot\hat{n})$.  This does, in fact, provide a decomposition of each effect $M^r_a$ in terms of the parent POVM $\{\Pi_{\hat{n}}\}_{\hat{n}}$ since
\begin{equation}
M^{r=1/2}_{a}=\mu_a(\mbb{I}+\frac{1}{2}\hat{m}_a\cdot\vec{\sigma})=\int_{\mc{S}} d\hat{n}2\mu_a \Theta(\hat{m}_a\cdot\hat{n})\Pi_{\hat{n}}.
\label{eq: rescaled respons}
\end{equation}
Unfortunately, this response function is not normalized; i.e, $\sum_a p(a|\{M^r_a\}_a,\hat{n})\ne 1$ for all $\hat{n}$.  To remedy this, we use the fact that the collection of (not necessarily normalized) response functions satisfying Eq. \eqref{eq: rescaled respons} is under constrained, 
and we judiciously search it to find a normalized one.  We show in detail below how this can be done, focusing first on three-outcome POVMs and then moving to four outcomes.  By the remark above, this covers all extremal qubit POVMs and therefore solves the full problem.  We note that the equivalence between PVMs and POVMs for the three-outcome case was previously shown in Ref. \cite{Werner2014} using different means.
\begin{figure}[t]
    \centering
\includegraphics[width=0.48\textwidth]{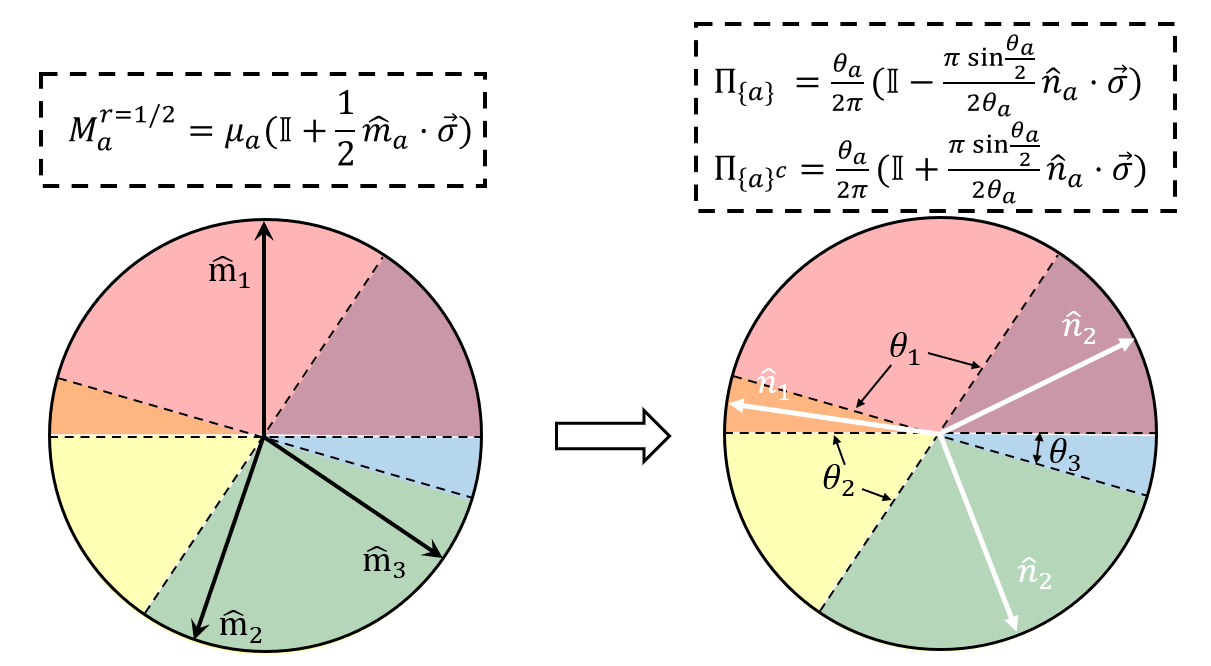}
    \caption{Top view on the plane defined by $\{\hat{m}_a\}_{a=1}^3$: The sphere is divided into 6 regions depending on $\Theta(\hat{m}_a\cdot\hat{n})$ and $\Pi_A$ can be obtained by integrating on the unit circle.  The analytic computation is described in the Method section.}
    \label{fig:partition}
\end{figure}
\begin{theorem}
Any noisy three-outcome POVM $\{M^{r=1/2}_a\}_{a=1}^3$ can be simulated by the POVM $\{\Pi_{\hat{n}}\}_{\hat{n}}$.
 \label{thm:3outcome}
\end{theorem}
A detailed proof is presented in the Methods section, but the key idea backing our method is coarse-graining the parent POVM $\{\Pi_{\hat{n}}\}_{\hat{n}}$ into a finite-outcome POVM $\{\Pi_A\}_A$ which is capable of simulating the given POVM $\{M^{r=1/2}_a\}_a$.  This implies that $\{\Pi_{\hat{n}}\}_{\hat{n}}$ is also capable of simulating $\{M^{r=1/2}_a\}_a$.
 \par
The coarse-grained POVM $\{\Pi_A\}_A$ is determined by which combinations of the Bloch vectors $\hat{m}_a$ are ``on'' for a given $\hat{n}$ (i.e. $\Theta(\hat{m}_a\cdot\hat{n})$ equaling one).  More precisely, we define
\begin{align}
\Pi_A&=\int_{\mc{S}} d\hat{n} \prod_{a\in A}\Theta(\hat{m}_a\cdot\hat{n})\prod_{a'\notin A}(1-\Theta(\hat{m}_{a'}\cdot\hat{n}))\Pi_{\hat{n}}
\label{eq:cg effects} 
\end{align}
with $A\subseteq \{1,2,3\}$ and $\Pi_{\emptyset}=\Pi_{\{1,2,3\}}=0$.   For the parent measurement $\{\Pi_A\}_A$, the starting (unnormalized) response function $p(a|A)$ for the $M^{r=1/2}_a$ defined in Eq.~\ref{eq: rescaled respons} is given by the values in following table:
 \begin{table}[h]
    \centering
    \begin{tabular}{c|c|c|c|c|c|c}
    \hline
     & $\Pi_{\{2,3\}}$  & $\Pi_{\{1,3\}}$ & $\Pi_{\{1,2\}}$  & $\Pi_{\{3\}}$ & $\Pi_{\{2\}}$ & $\Pi_{\{1\}}$     \\
    \hline
    $M^{r=1/2}_1$ & 0 &$2\mu_1$ & $2\mu_1$& 0 & 0& $2\mu_1$\\ 
    $M^{r=1/2}_2$ & $2\mu_2$ & 0 & $2\mu_2$& 0 & $2\mu_2$& 0\\ 
    $M^{r=1/2}_3$ &$2\mu_3$ &$2\mu_3$ & 0& $2\mu_3$ & 0& 0\\ 
    \hline
    \end{tabular}
    \caption{Unnormalized response function $p(a|A)$ for simulating $\{M_a\}_{a=1}^3$ with $\{\Pi_A\}_A$.}
    \label{tab:Response function}
\end{table}
\\
Due to the two completion relations, $\sum_a M^{r=1/2}_a=\mbb{I}$ and $\sum_A\Pi_A=\mbb{I}$, the six effects $\Pi_A$ are linear dependent {and satisfy 
\begin{equation}
    \sum_{a=1}^3 q_a (\Pi_{\{a\}}-\Pi_{\{a\}^c})=0,
    \label{eq: lp1}
\end{equation}
where $q_a:=1-2\mu_a$ and $A^c=\{1,2,3\}\setminus A$ is the set complement of $A$.  Moreover the spherical symmetry in the coarse-graining implies that for every $a,a'\in \{1,2,3\}$, $\Pi_{\{a\}}+\Pi_{\{a\}^c}$ has a vanishing Bloch vector and 
\begin{equation}
    \frac{1}{\alpha_a}(\Pi_{\{a\}}+\Pi_{\{a\}^c}) -\frac{1}{\alpha_{a'}}(\Pi_{\{a'\}}+\Pi_{\{a'\}^c})=0,
   \label{eq: lp2}
\end{equation}
where $\alpha_a:=\tr(\Pi_{\{a\}})/2=\tr(\Pi_{\{a\}^c})/2>0$.}
We can modify the response function $p(a|A)$ in Table \ref{tab:Response function} by adding/subtracting Eqns. \eqref{eq: lp1} and \eqref{eq: lp2}, thereby changing the weights of the different effects $\Pi_A$ while maintaining a simulation of the $M^{r=1/2}_a$.  In particular, a new normalized function $p'(a|A)$ is specified in Table \ref{tab:Response function1} where
\begin{table}[t]
    \centering
    \begin{tabular}{c|c|c|c|c|c|c}
    \hline
     & $\Pi_{\{2,3\}}$  & $\Pi_{\{1,3\}}$ & $\Pi_{\{1,2\}}$  & $\Pi_{\{3\}}$ & $\Pi_{\{2\}}$ & $\Pi_{\{1\}}$    \\
    \hline
    $M^{r=1/2}_1$ & 0 &$2\mu_1$ & $2\mu_1$& 0 & 0& $2\mu_1$\\ 
    $M^{r=1/2}_2$ & $2\mu_2-X$ & 0 & $1-2\mu_1$& 0 & $1$& $Y$\\ 
    $M^{r=1/2}_3$ &$2\mu_3-Y$ &$1-2\mu_1$ & 0& $1$ & 0& $X$\\ 
    \hline
    \end{tabular}
    \caption{Normalized response function $p'(a|A)$ for simulating $\{M_a\}_{a=1}^3$ with $\{\Pi_A\}_A$.}
    \label{tab:Response function1}
\end{table}
\begin{align}
X&=\frac{\alpha_1q_1-\alpha_2q_2+\alpha_3q_3}{2\alpha_1},&
Y&=\frac{\alpha_1q_1+\alpha_2q_2-\alpha_3q_3}{2\alpha_1}.\notag
\end{align}
For $p'(a|A)$ to be a well-defined response function, all the values in Table \ref{tab:Response function1} must be non-negative and each column must be normalized.  This can be verified by a straightforward calculation carried out in the Methods section.  
Therefore, we conclude that
\begin{equation}
    M_a^{r=1/2}=\sum_Ap'(a|A)\Pi_A,
\end{equation}
and so any arbitrary three-outcome noisy measurement $\{M_a^{r=1/2}=\mu_a(\mbb{I}+1/2\vec{m}_a\cdot\vec{\sigma})\}_{a=1}^3$ can be simulated by a common POVM $\{\Pi_{\hat{n}}\}_{\hat{n}}$ with a response function $p'(a|A)$ that depends on (i) the coeffecients $\mu_a$ and (ii) the coarse-graining of the $\{\Pi_{\hat{n}}\}_{\hat{n}}$ into the $\{\Pi_A\}_A$.

Note that Eq. \eqref{eq:cg effects} reduces to integration on the unit circle, which is geometrically depicted in Fig~\ref{fig:partition} and analytically described in the Methods section. We also emphasize a key property in the proof of Theorem \ref{thm:3outcome}.
\begin{observation}
To renormalize the response function of a three-outcome POVM $\{M^{r=1/2}_a\}_{a=1}^3$, it is sufficient to change the response function of two of its effects and leave the third one untouched. 
\label{obs:obs}
\end{observation}
Observation \ref{obs:obs} will be critical in extending Theorem \ref{thm:3outcome} to the four-outcome case.
\begin{theorem}
Any noisy four-outcome POVM $\{M^{r=1/2}_a\}_{a=1}^4$ can be simulated by the POVM $\{\Pi_{\hat{n}}\}_{\hat{n}}$.
 \label{thm:4outcome}
\end{theorem}
\begin{figure}[t]
    \centering
    \includegraphics[width=0.5\textwidth]{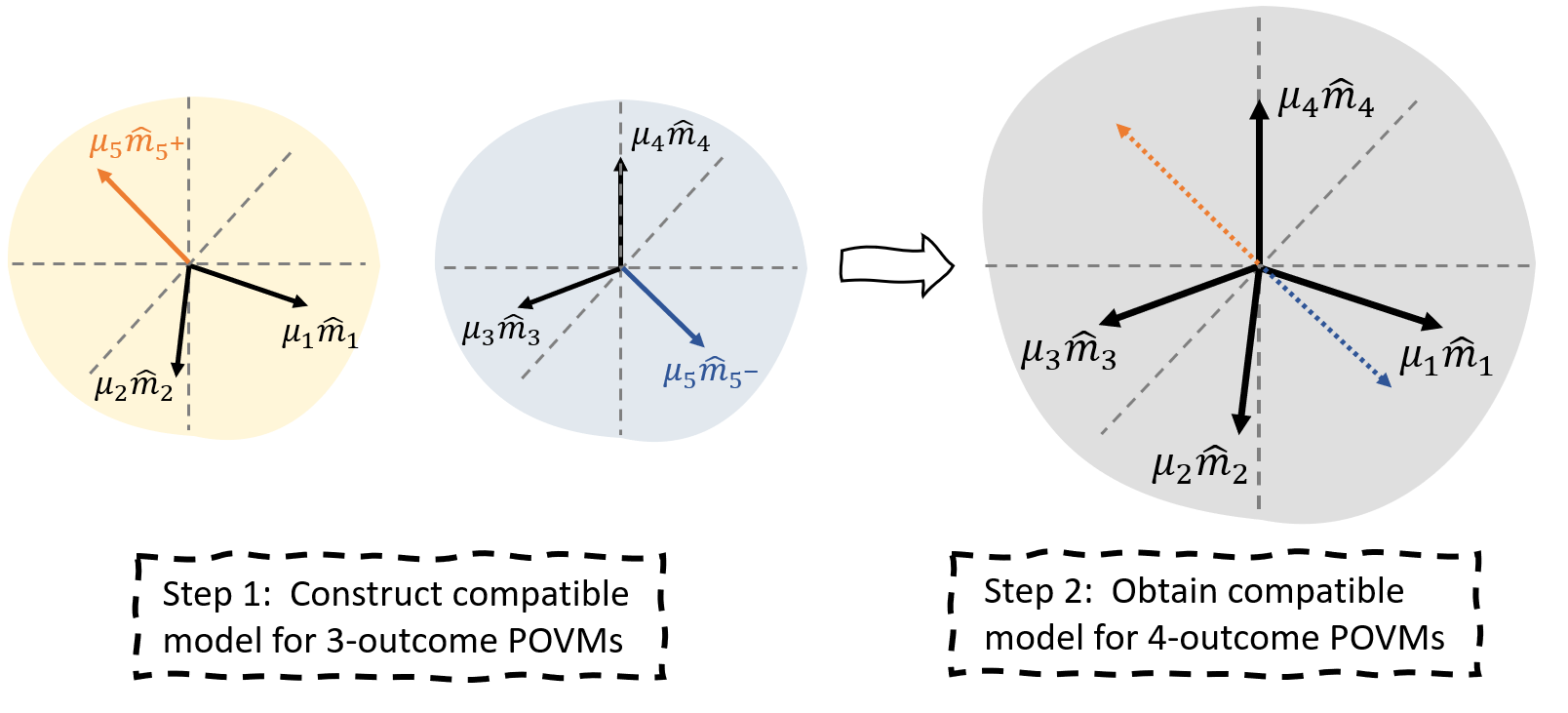}
    \caption{A schematic for constructing a compatible model for a four-outcome measurement $\{M^{r=1/2}_a\}_{a=1}^4$. Step 1: Construct compatible models for $\text{POVM}_+$ and $\text{POVM}_-$  without changing the response function of "pseudo effects" $M^{r=1/2}_{5^\pm}$; Step 2: Build a compatible model for the four-outcome measurement by combining the response function from step 1 and the coarse-grained POVM defined by all six vectors $\{\hat{m}_{a=1}^4\}\cap \{\hat{m}_{5^{\pm}}\}$.} 
    \label{fig:enter-label}
\end{figure}
\noindent The proof in the Methods section has two parts.  In the first part, we introduce two additional ``pseudo effects'' $M^{r=1/2}_{5^\pm}$ defined as
\begin{equation}
        M^{r=1/2}_{5^\pm}:={\mu_5}(\mbb{I}\pm \frac{1}{2}\hat{m}_5\cdot\vec{\sigma}),
\end{equation}
where $\hat{m}_5=-\frac{\mu_1\hat{m}_1+\mu_2\hat{m}_2}{|\mu_1\hat{m}_1+\mu_2\hat{m}_2|}$ and $\mu_5=|\mu_1\hat{m}_1+\mu_2\hat{m}_2|$.  The pseudo effects $M_{5^\pm}$ split the given POVM $\{M^{r=1/2}_a\}_{a=1}^4$ into two new POVMs
    \begin{align}
        \text{POVM}_+&=\frac{1}{\kappa_+}\{M^{r=1/2}_{5^+}, M^{r=1/2}_1,M^{r=1/2}_2\}\notag\\
        \text{POVM}_-&=\frac{1}{\kappa_-}\{M^{r=1/2}_{5^-}, M_3^{r=1/2},M_4^{r=1/2}\},
    \label{eq:POVM+-}
    \end{align}
    where $\kappa_+={\mu_1+\mu_2+\mu_5}$ and $\kappa_-={\mu_3+\mu_4+\mu_5}$.  One can easily verify that these are valid POVMs by the definitions of $(\mu_5,\hat{m}_5)$. 
    
Define the anti-parallel Bloch vectors $\hat{m}_{5^{\pm}}:=\pm\hat{m}_5$ along with outcome sets $\mc{O}_+=\{3,4,5^+\}$ and $\mc{O}_-=\{1,2,5^-\}$. By Theorem \ref{thm:3outcome}, $\text{POVM}_{\pm}$ can be simulated by a coarse-grained POVM $\{\Pi^{\pm}_{A_{\pm}}\}_{A_{\pm}}$ whose effects are
    \begin{align}
\Pi^{\pm}_{A_{\pm}}=\int_{\mc{S}} d\hat{n} \prod_{a\in A_{\pm}}\Theta(\hat{m}_a\cdot\hat{n})\prod_{a'\notin A_{\pm}}(1-\Theta(\hat{m}_{a'}\cdot\hat{n}))\Pi_{\hat{n}},
\label{eq:Akeffect}
\end{align}  
where $A_{\pm}\subset\mc{O}_{\pm}$. The compatible models can be explicitly written as
\begin{equation}
  \frac{1}{\kappa_{\pm}}  M_a^{r=1/2}=\sum_{A_{\pm}\subset \mc{O}_{\pm}}p_{\pm}(a|A_{\pm})\Pi^{\pm}_{A_{\pm}}\quad\text{for $a\in\mc{O}^{\pm}$}.
\end{equation}
Crucially, by Observation~\ref{obs:obs} and Table \ref{tab:Response function1}, the response function for $M_{5^{\pm}}^{r=1/2}/{\kappa_{\pm}}$ can be taken as $p_{\pm}(5^{\pm}|A_{\pm})=2\mu_5/\kappa_{\pm}$ if $5^{\pm}\in A_{\pm}$ and $p_{\pm}(5^{\pm}|A_{\pm})=0$ otherwise.
\par 
In the second part of our construction, we use the combined set of Bloch vectors $\{\hat{m}_a\}_{a=1}^4\cup\{+\hat{m}_5,-\hat{m}_5\}$ to define an 18-effect POVM $\{\Pi_A\}_A$ with
\begin{align}
\Pi_A=\int_{\mc{S}} d\hat{n} \prod_{a\in A}\Theta(\hat{m}_a\cdot\hat{n})\prod_{a'\notin A}(1-\Theta(\hat{m}_{a'}\cdot\hat{n}))\Pi_{\hat{n}},
\label{eq: 18effect}
\end{align}
where $A\subset \{1,2,3,4,5^+,5^-\}$. By comparing Eqns. \eqref{eq:Akeffect} and \eqref{eq: 18effect}, we see that the $\Pi_A$ provide a fine-graining of the $\Pi^{\pm}_{A_\pm}$ such that
\begin{align}
\Pi^{\pm}_{A_\pm}&=\sum_{A\cap\mc{O}_{\pm}=A_{\pm}} \Pi_A.
\end{align} 
Therefore the response functions $p(a|A_{\pm})$ defined for $\text{POVM}_{\pm}$ can be used to define a finer-grained response function, 
\begin{align}
q(a|A):=\kappa_{\pm }p_{\pm}(a|A\cap\mc{O}_{\pm})\quad \text{for $a\in\mc{O}_{\pm}$},
\end{align}
that provides a simulation $M_a=\sum_{A}q(a|A)\Pi_A$.  Verifying the normalization of $q(a|A)$ for arbitrary $A$ relies on the specific construction of $p_{\pm}(a|A)$ in the first step, which is carried out in the Methods section. 

Compared to the three-outcome case, a closed-form expression for the POVM $\{\Pi_A\}_A$ in Theorem \ref{thm:4outcome} cannot be obtained since the spherical integration in Eq.~\ref{eq: 18effect} is not analytically tractable.  Nevertheless, numerical integration can always be applied as we also discuss in the Methods section.  

\medskip

\textbf{Separating PVMs and POVMs in restricted quantum steering  --}
{
We now turn our attention to a slightly different problem.  In the quantum steering setting, suppose that Bob has a fixed state ensemble $\mathfrak{E}=(\rho_\lambda,p(\lambda))$ for his system.  We say that $\mf{E}$ can simulate Alice's measurement $\{M_a\}_a$ on bipartite state $\rho_{AB}$ if there exists a response function $p(a|\lambda)$ such that
\begin{equation}
\tr_A[(M_a\otimes\mbb{I})\rho_{AB}]=\int d\lambda p(\lambda) p(a|\lambda)\rho_\lambda.
\end{equation}
We are interested in how well a \textit{fixed} ensemble $\mf{E}$ can simulate Werner states under PVMs versus POVMs.  Let $r_{\text{PVM}}(\mf{E})$ (resp. $r_{\text{POVM}}(\mf{E})$) denote the largest $r$ such that $\mf{E}$ can simulate any PVM (resp. POVM) on $\rho_W(r)$.  Theorem \ref{thm:4outcome} says that $r_{\text{PVM}}(\mf{E})=r_{\text{POVM}}(\mf{E})=1/2$ when $\mf{E}$ is the ensemble of qubit pure states distributed uniformly on the Bloch sphere.  However, for a general ensemble $\mf{E}$ is it true that $r_{\text{PVM}}(\mf{E})=r_{\text{POVM}}(\mf{E})$?  We show that this is not always the case.  Our argument again uses the correspondence between LHS and compatibility models.  By Lemma \ref{lem:bound} the equivalent question to the one posed here is whether a fixed parent POVM $\{\Pi_\lambda\}_{\lambda}$ can always simulate noisy PVMs versus POVMs at the same noise threshold $r$.  The following proposition says that this is not always possible.}
\begin{proposition}
There exists a five-effect parent POVM $\{\Pi_\lambda\}_{\lambda=1}^5$ that can simulate all PVMs with radius $r\le 0.3714$, whereas a compatible model fails to exist for some three-outcome POVMs with radius $r> 0.3220$
\label{prop:Farka}
\end{proposition}  
\noindent The counterexample is constructed in the Supplementary Material, and Farka's Lemma is used to prove the non-existence of a compatible model.
\medskip

\textbf{Implications for Bell nonlocality --}  {Bell nonlocality is no weaker than steerability in the sense that every LHS model for one-sided measurements on a bipartite state $\rho_{AB}$ can be converted into a LHV model for two-sided measurements \cite{Wiseman2007, Jones2007}.  Such a connection has been implicitly used by Werner and Barrett \cite{Werner1989, Barrett2002} in deriving their original LHV models for the Werner state with radius $r=1/2$ and $r=5/12$ under PVMs and POVMs respectively.  Since these initial results, a breakthrough was made in the construction of LHV models under PVMs by relating the problem to finding upper bounds on Grothendieck's constant \cite{Acin2006,  Hirsch2017, designolle2023}.  This ensures that $\rho_W(r)$ is Bell local under PVMs whenever $r\leq\approx 0.6875$.   However, for general POVMs this method does not directly apply, and the best known locality bound under POVMs is $r\leq\approx 0.4517$, which has been derived by simulating noisy POVMs with PVMs \cite{Michal2017}.  This leaves a gap in the known locality range of $\rho_W(r)$ for PVMs versus POVMs.  Our Theorem~\ref{thm:4outcome} makes substantial progress toward closing that gap.
\begin{proposition}
There exists an LHV model for general measurements on the Werner state $\rho_W(r)$ when $r\le 1/2$.
\end{proposition}
 
\section{Conclusions}
In this paper, we have derived an exact bound for steering two-qubit Werner states under positive operator-valued measures and closed an open question in the literature \cite{IQOQI2017}.  Our method involves constructing an explicit LHS model for general measurements that requires an unbounded amount of shared randomness for Alice and Bob. However, we also show that in a more restricted quantum steering scenario, PVMs and POVMs can play different roles.

There are a number of open problems related to the subject of our work. First, it is natural to ask whether our results in deriving LHS (or compatibility) models under general measurements can be extended from qubits to systems with arbitrary dimensions. The second interesting problem involves exploring the distinction between PVMs and POVMs in other restricted quantum steering or Bell nonlocality scenarios. For instance, one could consider quantum steering and Bell nonlocality when using finite amounts of shared randomness or when $\rho_{AB}$ is a state with no symmetry.  Finally, even among the Werner family of states, our methods are not strong enough to construct LHV models $r>1/2$.  Consequently, when $r>1/2$ it remains an open question whether PVMs and POVMs are equally powerful for realizing Bell nonlocality.  

\section{Methods}
\subsection{Proof of Theorem 1:} 
\begin{proof}
As explained in Theorem~\ref{thm:3outcome} of the main text, to compute a normalized response function, a coarse-grained POVM $\{\Pi_A\}_A$ is constructed from $\{\Pi_{\hat{n}}=\frac{1}{4\pi}(\mbb{I}+\hat{n}\cdot\vec{\sigma})\}$ based on the given POVM effects
    $$M_a^{r=\frac{1}{2}}=\mu_a(\mbb{I}+\frac{1}{2}\hat{m}_a\cdot\vec{\sigma}), $$
where $\frac{1}{2}\ge \mu_a\ge 0$, $\sum_a\mu_a=0$ and $\sum_a\mu_a\vec{m}_a=\vec{0}$.  To simplify notation going forward, we will remove the superscript $r=1/2$ and write $\wt{M}_a:=M_a^{r=1/2}$.

An analytical expression of the six-effect POVM $\{\Pi_A\}$ defined in Eq. \eqref{eq:cg effects} can be explicitly computed by integrating over six different regions shown in Fig~\ref{fig:partition} (defined by coplanar vector $\{\hat{m}_a\}$).  For example,
\begin{align}
\Pi_{\{1\}}&=\frac{\theta_{1}}{2\pi}(\mbb{I}-\frac{\pi\sin\frac{\theta_{1}}{2}}{2\theta_{1}}\hat{n}_{1}\cdot\vec{\sigma})\\
\Pi_{\{1\}^c}&=\frac{\theta_{1}}{2\pi}(\mbb{I}+\frac{\pi\sin\frac{\theta_{1}}{2}}{2\theta_{1}}\hat{n}_{1}\cdot\vec{\sigma}),
\label{eq:anl POVM}
\end{align}
where $\hat{n}_{_{1}}=\frac{\hat{m}_2+\hat{m}_3}{|\hat{m}_2+\hat{m}_3|}$ and $\theta_{1}=\arccos(-\hat{m}_2\cdot\hat{m}_3)$, and likewise for the other POVM elements.
\par 
The combination of linear dependent relation in Eq.~\ref{eq: lp1} and Eq.~\ref{eq: lp2} give rise to constraints:
\begin{align}
    &Y\Pi_{\{1\}}+q_2\Pi_{\{2\}}=q_3\Pi_{{\{3\}^c}}+X\Pi_{{\{1\}}^c}
    \notag\\   &X\Pi_{\{1\}}+q_3\Pi_{\{3\}}=q_2\Pi_{{\{2\}^c}}+Y\Pi_{{\{1\}}^c},
\end{align}
where
\begin{align}
X=\frac{\alpha_1q_1-\alpha_2q_2+\alpha_3q_3}{2\alpha_1}~~~~
Y=\frac{\alpha_1q_1+\alpha_2q_2-\alpha_3q_3}{2\alpha_1},
\label{eq: lp3}
\end{align}
and $\alpha_i=\tr(\Pi_{\{i\}})/2={\theta_i}/{2\pi}$.  The positivity of $X$ and $Y$ can be verified by rewriting the linear dependence in Eq.~\ref{eq: lp1} with Eq.~\ref{eq:anl POVM} as:
\begin{equation}
    \sum_{a=1}^3 q_a (\Pi_{{\{a\}}^c}-\Pi_{\{a\}})=0\Longrightarrow    \sum_{a=1}^3 q_a\sin(\frac{\theta_{a}}{2})\hat{n}_{a}=\vec{0},
\end{equation}
\noindent 
Using identity $\hat{n}_1\cdot\hat{n}_3=\cos(\frac{\theta_1+\theta_3}{2})$, $\hat{n}_2\cdot\hat{n}_3=\cos(\frac{\theta_2+\theta_3}{2})$ we can now write down the chain of inequalities:
\begin{align}
   \theta_1q_1+\theta_2q_2 &\ge 2[q_1\sin(\frac{\theta_1}{2})+q_2\sin(\frac{\theta_2}{2})]\notag \\
 &\ge 2\frac{q_1\sin(\frac{\theta_1}{2})\cos(\frac{\theta_1+\theta_3}{2})+q_2\sin(\frac{\theta_2}{2})\cos(\frac{\theta_2+\theta_3}{2})}{\cos(\frac{\theta_3}{2})}
   \notag \\
   & = 2\frac{\sin(\frac{\theta_3}{2})q_3}{\cos(\frac{\theta_3}{2})}
   \notag = 2q_3\tan(\frac{\theta_3}{2})\notag\ge\theta_3q_3  
\end{align}
where in the first and last inequalities we use $\tan(x)\ge x\ge\sin(x)$ on $[0,\pi/2)$ and in the second inequality we notice that $\cos(x)$ is decreasing on $[0,\pi]$. Therefore, one can conclude that $Y\ge 0$ and similarly $X\ge 0$. Moreover, since $X+Y=q_1=1-2\mu_1$, we have
\begin{equation}
    0\le X,Y\le 1-2\mu_1.
\end{equation}
Now, starting with the initial response function given in Table \ref{tab:Response function}, it is safe to add the linear dependent relation in Eq. \ref{eq: lp3} as:
\begin{align}
\wt{M}_2&=2\mu_2\Pi_{\{2\}}+2\mu_{2}\Pi_{\{1\}^c}+2\mu_{3}\Pi_{\{3\}^c}\notag\\
&+Y\Pi_{\{1\}}+q_2\Pi_{\{2\}}-q_3\Pi_{{\{3\}^c}}-X\Pi_{{\{1\}}^c}\notag\\ 
&=Y\Pi_{\{1\}}+\Pi_{\{2\}}+(2\mu_2-X)\Pi_{\{1\}^c}+(1-2\mu_1)\Pi_{\{3\}^c} \notag \\
\wt{M}_3&=2\mu_3\Pi_{\{3\}}+2\mu_{3}\Pi_{\{1\}^c}+2\mu_{3}\Pi_{\{2\}^c}\notag \\
&+X\Pi_{\{1\}}+q_3\Pi_{\{3\}}-q_2\Pi_{{\{2\}^c}}-Y\Pi_{{\{1\}}^c}
\notag\\ 
&=X\Pi_{\{1\}}+\Pi_{\{3\}}+(2\mu_3-Y)\Pi_{{\{1\}
^c}}+(1-2\mu_1)\Pi_{{\{2\}^c}}
\notag
\end{align}
where $2\mu_2-X\ge 2\mu_2+2\mu_1-1=1-2\mu_3\ge 0 $ and  $2\mu_3-Y\ge 2\mu_3+2\mu_1-1=1-2\mu_2\ge 0 $.
\par 
Therefore, the new response function $p'(a|A)$ for $M_a$ with $a=2$ and $a=3$ are well-defined and are normalized.  For $a=1$, the response function remains the same and
\begin{equation}
    \wt{M}_1=2\mu_1\Pi_{\{1\}}+2\mu_1\Pi_{\{2\}^c}+2\mu_1\Pi_{\{3\}^c}
\end{equation}

\end{proof}
\subsection{Proof of Theorem 2:}
\begin{proof}
We continue to use the notation $\wt{M}_a:=M_a^{r=1/2}$.  First, by introducing a pair of "pseudo effects" $ \wt{M}_{5^{\pm}}={\mu_5}(\mbb{I}\pm \frac{1}{2}\hat{m}_5\cdot\vec{\sigma})$, with $\hat{m}_5=-\frac{\mu_1\hat{m}_1+\mu_2\hat{m}_2}{|\mu_1\hat{m}_1+\mu_2\hat{m}_2|}$ and $\mu_5=|\mu_1\hat{m}_1+\mu_2\hat{m}_2|$. It is easy to verify that $\text{POVM}_{\pm}$ defined in Eq.~\ref{eq:POVM+-} are valid POVMs.  Also, note that 
$$\mu_5\hat{m}=\mu_3\hat{m}_3+\mu_4\hat{m}_4=-\mu_1\hat{m}_1-\mu_2\hat{m}_2,$$
and so $\kappa_{\pm}\le \sum_{a=1}^{4} \mu_a=1$ by the triangle inequality.
\par Both $\text{POVM}_+$ and $\text{POVM}_-$ have three effects; thus the construction in Theorem~\ref{thm:3outcome} allows us to build coarse-grained POVMs $\{\Pi^{\pm}_{A_{\pm}}\}$ and normalized response functions $p_{\pm}(a|A_{\pm})$ to simulate  $\text{POVM}_\pm$.  Explicitly, we have
\begin{align}
\Pi^{\pm}_{A_{\pm}}=\int_{\mc{S}} d\hat{n} \prod_{a\in A_{\pm}}\Theta(\hat{m}_a\cdot\hat{n})\prod_{a'\notin A_{\pm}}(1-\Theta(\hat{m}_{a'}\cdot\hat{n}))\Pi_{\hat{n}},
\label{eq:6effectpm}
\end{align} 
with $A_{\pm}\subset\mc{O}_{\pm}$
and $\mc{O}_+=\{5^+,1,2\}$ and $\mc{O}_-=\{5^-,3,4\}$.
For $\text{POVM}_+$, the response function $p_+(a|A_+)$ is given by
\begin{table}[h]
    \centering
    \begin{tabular}{c|c|c|c|c|c|c}
    \hline
     & $\Pi^+_{\{1,2\}}$  & $\Pi^+_{\{2,5^+\}}$  & $\Pi^+_{\{1,5^+\}}$ & $\Pi^+_{\{2\}}$ & $\Pi^+_{\{1\}}$ & $\Pi^+_{\{5^+\}}$     \\
    \hline
    $\frac{1}{\kappa_+}\wt{M}_{5^{+}} $ & 0 &$\frac{2}{\kappa_+}\mu_5$ & $\frac{2}{\kappa_+}\mu_5$& 0 & 0& $\frac{2}{\kappa_+}\mu_5$\\ 
    $\frac{1}{\kappa_+}\wt{M}_{1}$ & $\frac{2}{\kappa_+}\mu_1-X_+$ & 0 & $1-\frac{2}{\kappa_+}\mu_5$& 0 & $1$& $Y_+$\\ 
    $\frac{1}{\kappa_+}\wt{M}_{2}$ &$\frac{2}{\kappa_+}\mu_2-Y_+$ &$1-\frac{2}{\kappa_+}\mu_5$ & 0& $1$ & 0& $X_+$\\ 
    \hline
    \end{tabular}
    \label{tab:Response function3}
    \caption{}
\end{table}\\
where
\begin{align}
X_+=\frac{\beta_{5^+}q_{5^+}+\beta_1q_1-\beta_2q_2}{2\beta_{5^+}},~~~~Y_+=\frac{\beta_{5^+}q_{5^+}-\beta_1q_1+\beta_2q_2}{2\beta_{5^+}}
\end{align}
with $\beta_a=\frac{\tr{\Pi^+_{\{a\}}}}{2}$ and $q_a=1-\frac{2}{\kappa_+}\mu_a$ $a\in\mc{O}_+$.

For $\text{POVM}_-$, we likewise have $p_-(a|A_-)$ as
\begin{table}[h]
    \centering
    \begin{tabular}{c|c|c|c|c|c|c}
    \hline
     & $\Pi^-_{\{3,4\}}$  & $\Pi^-_{\{4,5^-\}}$ & $\Pi^-_{\{3,5^-\}}$ & $\Pi^-_{\{4\}}$ & $\Pi^-_{\{3\}}$ & $\Pi^-_{\{5^-\}}$     \\
    \hline
    $\frac{1}{\kappa_-}\wt{M}_{5^-}$ & 0 &$\frac{2}{\kappa_-}\mu_5$ & $\frac{2}{\kappa_-}\mu_5$& 0 & 0& $\frac{2}{\kappa_-}\mu_5$\\ 
    $\frac{1}{\kappa_-}\wt{M}_{3}$ & $\frac{2}{\kappa_-}\mu_3-X_-$ & 0 & $1-\frac{2}{\kappa_-}\mu_5$& 0 & $1$ & $Y_-$\\ 
    $\frac{1}{\kappa_-}\wt{M}_{4}$ &$\frac{2}{\kappa_-}\mu_4-Y_-$ &$1-\frac{2}{\kappa_-}\mu_5$ & 0& $1$ & 0& $X_-$\\ 
    \hline
    \end{tabular}
    \label{tab:Response function2}
    \caption{}
\end{table} 
\\where 
\begin{align}
    X_-&=\frac{\gamma_{5^-}q_{5^-}+\gamma_3q_3-\gamma_4q_4}{2\gamma_{5^-}}~~~~
Y_-=\frac{\gamma_{5^-}q_{5^-}-\gamma_3q_3+\gamma_4q_4}{2\gamma_{5^-}}
\end{align}
with $\gamma_a=\frac{\tr{\Pi^-_{\{a\}}}}{2}$  and $q_a=1-\frac{2}{\kappa_-}\mu_a$ for for $a\in\mc{O}_-$\par 

These constructions follow from Table \ref{tab:Response function1} and ultimately from the fact that the values for $p_{\pm}(5^{\pm}|A_{\pm})$ can remain unchanged during the renormalization (Observation \ref{obs:obs}).\par 

The second part of the proof relies on the construction of an 18-effect POVM $\{\Pi_A\}_A$ defined by the set of vector $\{\hat{m}_a\}_{a=1}^4\cap\{\pm \hat{m}_5\}$
\begin{align}
\Pi_A=\int_{\mc{S}} d\hat{n} \prod_{a\in A}\Theta(\hat{m}_a\cdot\hat{n})\prod_{a'\notin A}(1-\Theta(\hat{m}_{a'}\cdot\hat{n}))\Pi_{\hat{n}}
\label{eq: 18effectms}
\end{align}
where $A\subset \{1,2,3,4,5^+,5^-\}$. Note that $\Pi_A=0$ if $A$ or $A^c$ enumerates a set of Bloch vectors containing the origin in its convex hull.  In particular, $\Pi_A=0$ if either $\{5^+,5^-\}\subset A$ or $\{5^+,5^-\}\subset A^c$.  The set of nonzero POVM $\Pi_A$ are
\begin{equation}
\begin{bmatrix}
  \Pi_{\{1,5^-\}} & \Pi_{\{2,5^-\}} & \Pi_{\{1,3,5^-\}} & \Pi_{\{2,3,5^-\}} & \Pi_{\{1,2,5^-\}} \\   
  \Pi_{\{3,5^+\}} & \Pi_{\{4,5^+\}} &   \Pi_{\{1,4,5^+\}} & \Pi_{\{2,4,5^+\}}  & \Pi_{\{1,3,5^+\}} \\
   \Pi_{\{1,4,5^-\}} & \Pi_{\{2,4,5^-\}}& \Pi_{\{1,2,4,5^-\}} & \Pi_{\{1,2,3,5^-\}}  \\
    \Pi_{\{2,3,5^+\}} & \Pi_{\{3,4,5^+\}} & \Pi_{\{1,3,4,5^+\}}  & \Pi_{\{2,3,4,5^+\}}
\end{bmatrix}
\label{eq:18nonzero}
\end{equation}
By comparing Eqns. \eqref{eq:6effectpm} and \eqref{eq: 18effectms}, we see that the $\Pi_A$ provide a fine-graining of the $\Pi^{\pm}_{A_\pm}$ such that
\begin{align}
\Pi^{\pm}_{A_\pm}&=\sum_{A\cap\mc{O}_{\pm}=A_{\pm}} \Pi_A.
\end{align} 
In more detail,
\begin{equation}
\begin{split}
\Pi^+_{\{1\}}&=\Pi_{\{1,5^-\}}+\Pi_{\{1,4,5^-\}}+\Pi_{\{1,3,5^-\}}\\
\Pi^+_{\{2\}}&=\Pi_{\{2,5^-\}}+\Pi_{\{2,4,5^-\}}+\Pi_{\{2,3,5^-\}}\\
\Pi^+_{\{5^+\}}&=\Pi_{\{4,5^+\}}+\Pi_{\{3,5^+\}}+\Pi_{\{3,4,5^+\}} \\
\Pi^+_{\{1,2\}}&=\Pi_{\{1,2,5^-\}}+\Pi_{\{1,2,3,5^-\}}+\Pi_{\{1,2,4,5^-\}} \\
\Pi^+_{\{1,5^+\}}&=\Pi_{\{1,3,5^+\}}+\Pi_{\{1,4,5^+\}}+\Pi_{\{1,3,4,5^+\}} \\
\Pi^+_{\{2,5^+\}}&=\Pi_{\{2,3,5^+\}}+\Pi_{\{2,4,5^+\}}+\Pi_{\{2,3,4,5^+\}} 
\end{split}
\end{equation}
and 
\begin{equation}
\begin{split}
\Pi^-_{\{3\}}&=\Pi_{\{3,5^+\}}+\Pi_{\{1,3,5^+\}}+\Pi_{\{2,3,5^+\}}\\
\Pi^-_{\{4\}}&=\Pi_{\{4,5^+\}}+\Pi_{\{1,4,5^+\}}+\Pi_{\{2,4,5^+\}}\\
\Pi^-_{\{5^-\}}&=\Pi_{\{1,5^-\}}+\Pi_{\{2,5^-\}}+\Pi_{\{1,2,5^-\}}\\
\Pi^-_{\{3,4\}}&=\Pi_{\{3,4,5^+\}}+\Pi_{\{1,3,4,5^+\}}+\Pi_{\{2,3,4,5^+\}}\\
\Pi^-_{\{3,5^-\}}&=\Pi_{\{2,3,5^-\}}+\Pi_{\{1,3,5^-\}}+\Pi_{\{1,2,3,5^-\}}\\
\Pi^-_{\{4,5^-\}}&=\Pi_{\{2,4,5^-\}}+\Pi_{\{1,4,5^-\}}+\Pi_{\{1,2,4,5^-\}}\\
\end{split} 
\end{equation}
Substituting this into the simulations of $\text{POVM}_{\pm}$ using Table~\ref{tab:Response function2} and Table~\ref{tab:Response function3} 
yields \begin{equation}
    M_a=\sum_{A}q(a|A)\Pi_A,
\end{equation}
where for nonzero $\Pi_A$ listed in Eq.~\ref{eq:18nonzero} we have
\begin{align}
q(a|A):=\kappa_{\pm}p_{\pm}(a|A\cap\mc{O}_{\pm})\quad \text{for $a\in\mc{O}_{\pm}$}.
\end{align}
Since $0\le \kappa_{\pm}\le 1$, these response functions are positive and bounded by one. Moreover, the normalization for each $A$ can be checked by noticing that
\begin{align}
\sum_{a=1}^4q(a|A)+q(5^+|A)+q(5^-|A)=\kappa_++\kappa_-=1+2\mu_5.\notag
\end{align}
Since $\Pi_A$ will be nonzero only if either $5^+\in A\cap \mc{O}_+$ or $5^-\in A\cap\mc{O}_-$ but not both, one has $q(5^+|A)+q(5^-|A)=2\mu_5$.  Comparing this with the previous equation shows that $\sum_{a=1}^4q(a|A)=1$, as desired.
\end{proof}
\subsection{Disputing compatible model with Farka's lemma}
We continue to use the notation $\wt{M}_a:=M_a^{r=1/2}$.  The condition for simulating children measurement $\{\wt{M}_a\}\}_a^m$ with measurement $\{\Pi_i\}_i^n$ can be expressed as a set of linear equations:
\begin{align}
    \sum_i x_{a|i} \Pi_i=\wt{M}_{a}~~~~\text{for $a=1,\cdots m$} \notag \\ 
    \sum_a x_{a|i}=1 ~~~~\text{for $i=1,\cdots n$}.
    \label{eq: compatible equation}
\end{align}
By writing qubit measurements as a vector in $\mbb{R}^4$, i.e., 
\begin{align}
    \wt{M}_a=\mu_a(\mbb{I}+1/2\hat{m}_a\cdot\vec{\sigma})
    &\Longrightarrow \vec{m}_a=(\mu_a,1/2\mu_a\hat{m}_a)^T\notag \\
    \Pi_i=\alpha_i(\mbb{I}+\hat{n}_i\cdot\vec{\sigma}) &\Longrightarrow \vec{\pi}=(\alpha_i,\alpha_i\vec{n}_i)^T,
\end{align}
the set of equations in Eq.~\ref{eq: compatible equation}} becomes  
$A\cdot \boldsymbol{x}=\boldsymbol{b}$ with $\boldsymbol{x}\geq 0$ and
\setcounter{MaxMatrixCols}{20}
\begin{align}
    A&=\begin{pmatrix} \vec{\pi}_1&\cdots&\vec{\pi}_n&0&\cdots&0&0&0&0\\
& \ddots &&& \ddots &&& \ddots &\\ 0&0&0&0&\cdots&0&\vec{\pi}_1&\cdots&\vec{\pi}_n\\
    1&\cdots  &0&1&\cdots    &0&1&\cdots  &0\\
    &\ddots  &&&\ddots  &&&\ddots  &\\
    0&\cdots  &1&0&\cdots  &1& 0&\cdots  &1\end{pmatrix},  \notag\\
    \boldsymbol{x}&=(x_{1|1},\cdots,x_{1|n} ,x_{2|n},\cdots, x_{m|n})^T, \notag\\
    {b}&=(\vec{m}_1,\cdots,\vec{m}_m, 1,\cdots, 1)^T.
    \label{eq: linear equations}
\end{align}
\begin{lemma}[Farkas' Lemma] 
\begin{equation}
  \nexists \boldsymbol{x}\ge \boldsymbol{0} \text{ s.t. } A\cdot \boldsymbol{x}=\boldsymbol{b} \Longleftrightarrow \exists y \text{ s.t. } A^T\cdot \boldsymbol{y}\ge \boldsymbol{0} \text{ and }\boldsymbol{b}^T\cdot \boldsymbol{y}<0
\end{equation}
where $\boldsymbol{y}\in\mbb{R}^{4m+n}$ as $y=(\vec{y}_1,\cdots,\vec{y}_m,z_1,\cdots,z_n)^T$. 
\end{lemma}
Therefore, to show the non-existence of a set of positive, normalized response function $\{x_{a|i}\}_{i,a}$ for simulating $\{\wt{M}_a\}$ with $\{\Pi_i\}$, it suffices to find a vector $\boldsymbol{y}$ such that $A^T\cdot \boldsymbol{y}\ge \boldsymbol{0} \text{ and }\boldsymbol{b}^T\cdot \boldsymbol{y}<0$ holds at the same time. \par 

In the Supplementary Material, we will use Farka's lemma to dispute the existence of a compatible model with a restricted setting (fixed finite-outcome POVM) as we state in Prop.~\ref{prop:Farka}.
\subsection{Numerical methods}
As mentioned in the main text, compared to the three-outcome measurements, a closed-form expression for the POVM $\{\Pi_A\}_A$ given in Eq.~\ref{eq: 18effect} is not analytically tractable. Here we briefly introduce the spherical numerical integration methods we used in our computational appendix \cite{zhang2023gita}. \par 
Integration in Eq.~\ref{eq: 18effect} can be rewritten as:
\begin{align}
\hat{\pi}_A=\int_{\mc{S}} d\hat{n} \prod_{a\in A}\Theta(\hat{m}_a\cdot\hat{n})\prod_{a'\notin A}(1-\Theta(\hat{m}_{a'}\cdot\hat{n}))\hat{\pi}_{\hat{n}}
\label{eq: 18effect-4vector}
\end{align}
where $\hat{\pi}_{\hat{n}}=1/4\pi(1,\hat{n})^T$ and $\hat{\pi}_A$ is the 4-vector representation of $\Pi_A$. \par 

To perform the spherical integration above, we use the 131th order Lebedev quadrature for the numerical computation \cite{Lebedev1999}, and the integration can be approximated as:
\begin{align}
\hat{\pi}_A=4\pi\sum_{i=1}^N \omega_i \prod_{a\in A}\Theta(\hat{m}_a\cdot\hat{l}_i)\prod_{a'\notin A}(1-\Theta(\hat{m}_{a'}\cdot\hat{l}_i))\hat{\pi}_{\hat{l}_i}
\label{eq: Lebedev}
\end{align}
where $\omega_i, \hat{l}_i$ are the so-called Lebedev weights and Lebedev grid. \par 
\section{Brute force linear solver}
A query that might be raised is whether there exists a compatible model for $\{M_a^{r=1/2}\}_{a=1}^4$ that can be constructed with a $14$-effect POVM, defined by:
\begin{align}
\Pi_A=\int_{\mc{S}} d\hat{n} \prod_{a\in A}\Theta(\hat{m}_a\cdot\hat{n})\prod_{a'\notin A}(1-\Theta(\hat{m}_{a'}\cdot\hat{n}))\Pi_{\hat{n}}
\label{eq: 14-effect}
\end{align}
where $A\subset \{1,2,3,4\}$ with $\Pi_{\emptyset}=\Pi_{\{1,2,3,4\}}=0$.\par 
It is important to note that, in general, the integration mentioned above cannot be computed analytically. This stands in contrast to cases where $\{\hat{m}_a\}_{a=1}^3$ lie in the same plane. Besides, the construction of response functions relies on some knowledge of these effects, and due to this complexity, performing an analytical treatment for such coarse-graining and simulation is not straightforward.

Nonetheless, a numerical approach remains viable. This involves solving the linear equations $A \cdot \boldsymbol{x} = \boldsymbol{b}$ as defined in Eq.~\ref{eq: linear equations}. The solution is subject to the constraint $\boldsymbol{x} \ge 0$, which can be achieved through a linear programming procedure outlined below:
\begin{align}
    \min~~ & 0 \notag\\
    \text{s.t.}~~& A\cdot \boldsymbol{x}=\boldsymbol{b}\notag \\
     & \boldsymbol{x}\ge \boldsymbol{0}
\end{align}
and strong numerical evidence is provided in our numerical appendix \cite{zhang2023gita} showing that there is no example of a POVM that cannot be simulated with a $14$-effect POVM having the coarse-graining defined in Eq.~\ref{eq: 14-effect}.  
\bibliography{ref1}
\newpage
\onecolumngrid
\section{Supplementary material}
Here we provide more detailed proofs and construction for some of the results made in the main text. 
\subsection{Construction of coarse-grained POVM $\{\Pi_A\}$ for simulating four-outcome measurement $\{M_a^{r=\frac{1}{2}}\}$}
Given a noisy four-outcome children POVM $\wt{M}_a:=M_a^{r=\frac{1}{2}}=\mu_a(\mbb{I}+1/2\hat{m}_a\cdot\vec{\sigma})$, a pair of pseudo effects $\wt{M}_{5^{\pm}}$ is defined to help with the construction of a coarse-grained POVM:
    \begin{equation}
        \wt{M}_{5^{\pm}}={\mu_5}(\mbb{I}\pm \frac{1}{2}\hat{m}_5\cdot\vec{\sigma})
    \end{equation}
where $\hat{m}_5=-\frac{\mu_1\hat{m}_1+\mu_2\hat{m}_2}{|\mu_1\hat{m}_1+\mu_2\hat{m}_2|}$ and $\mu_5=|\mu_1\hat{m}_1+\mu_2\hat{m}_2|$.
The corresponding $18$-effect POVM $\{\Pi_A\}_A$ defined by the `on' and `off' of vectors $\{\hat{m}_a\}_{a=1}^4\cap\{\pm\hat{m}_5\}$ can then be computed as:
\begin{align}
\Pi_A=\int_{\mc{S}} d\hat{n} \prod_{a\in A}\Theta(\hat{m}_a\cdot\hat{n})\prod_{a'\notin A}(1-\Theta(\hat{m}_{a'}\cdot\hat{n}))\Pi_{\hat{n}}
\label{eq:18effects}
\end{align}
where $A\subset\{1,2,3,4,5^+,5^-\}$. In addition, two extra six-effect POVMs $\{\Pi^{\pm}_{A_{\pm}}\}$ are also defined for the characterization of the response function $p(a|A)$ given by
  \begin{align}
\Pi^{\pm}_{A_{\pm}}=\int_{\mc{S}} d\hat{n} \prod_{a\in A_{\pm}}\Theta(\hat{m}_a\cdot\hat{n})\prod_{a'\notin A_{\pm}}(1-\Theta(\hat{m}_{a'}\cdot\hat{n}))\Pi_{\hat{n}}
\end{align}  
with $A_{\pm}\subset\mc{O}_{\pm}$, where $\mc{O_+}=\{5^+,1,2\}$ and $\mc{O_-}=\{5^-,3,4\}$

 The (normalized) response function of this $18$-effect $\{\Pi_A\}_A$ for the simulation of $\{\wt{M}_a\}_{a=1}^4$ is summarized in this table
\begin{table}[h]
\centering
\begin{adjustbox}{max width=\textwidth}
    \begin{tabular}{c|c|c|c|c|c|c|c|c|c|c|c|c|c|c|c|c|c|c}
    \hline
 & $\Pi_{\{1,5^-\}}$ & $\Pi_{\{2,5^-\}}$ &  $\Pi_{\{1,2,5^-\}}$ &$\Pi_{\{1,3,5^-\}}$ 
 & $\Pi_{\{1,4,5^-\}}$ &  $\Pi_{\{2,3,5^-\}}$ & $\Pi_{\{2,4,5^-\}}$& $\Pi_{\{3,5^+\}}$ 
 & $\Pi_{\{4,5^+\}}$ &  $\Pi_{\{1,2,3,5^-\}}$ & $\Pi_{\{1,2,4,5^-\}}$  &  $\Pi_{\{1,3,5^+\}}$ 
 & $\Pi_{\{1,4,5^+\}}$ &$\Pi_{\{2,3,5^+\}}$ &  $\Pi_{\{2,4,5^+\}}$ & $\Pi_{\{3,4,5^+\}}$ 
 &  $\Pi_{\{1,3,4,5^+\}}$ & $\Pi_{\{2,3,4,5^+\}}$  \\
    \hline
    $\wt{M}_1$ & $\kappa_+$ & 0 &$2\mu_1-X_+\kappa_+$ & $\kappa_+$
                          & $\kappa_+$ & 0 &0 &$Y_+\kappa_+$ 
                          & $Y_+\kappa_+$ &  $2\mu_1-X_+\kappa_+$ &  $2\mu_1-X_+\kappa_+$ & $\kappa_+-2\mu_5$
                          & $\kappa_+-2\mu_5$ & 0 & 0 &$Y_+\kappa_+$ 
                          &$\kappa_+-2\mu_5$ & 0\\
    $\wt{M}_2$ & 0 & $\kappa_+$& $2\mu_2-Y_+\kappa_+$ & 0
                          & 0 & $\kappa_+$ &$\kappa_+$ &$X_+\kappa_+$  
                          & $X_+\kappa_+$  &  $2\mu_2-Y_+\kappa_+$ &  $2\mu_2-Y_+\kappa_+$ & 0
                          & 0& $\kappa_+-2\mu_5$ & $\kappa_+-2\mu_5$ &$X_+\kappa_+$ 
                          &0 & $\kappa_+-2\mu_5$\\
    $\wt{M}_3$ & $Y_-\kappa_-$ & $Y_-\kappa_-$ & $Y_-\kappa_-$ & $\kappa_--2\mu_5$
                          & 0 &  $\kappa_--2\mu_5$ &0 &$\kappa_-$
                          & 0 &   $\kappa_--2\mu_5$ &  0& $\kappa_-$
                          & 0& $\kappa_-$ & 0 &$2\mu_3-X_-\kappa_-$
                          &$2\mu_3-X_-\kappa_-$ & $2\mu_3-X_-\kappa_-$\\
    $\wt{M}_4$ & $X_-\kappa_-$ & $X_-\kappa_-$ & $X_-\kappa_-$ & 0
                          & $\kappa_--2\mu_5$ & 0 &$\kappa_--2\mu_5$ &0 
                          & $\kappa_-$ &  0& $\kappa_--2\mu_5$ & 0
                          & $\kappa_-$& 0 & $\kappa_-$ &$2\mu_4-Y_-\kappa_-$
                          &$2\mu_4-Y_-\kappa_-$& $2\mu_4-Y_-\kappa_-$\\
    \hline
    \end{tabular}
    \end{adjustbox}
    \caption{Response function $p(a|A)$ for simulating $\{\wt{M}_a\}_{a=1}^4$ with $18$-effect POVM $\{\Pi_A\}_A$.}
    \label{tab:Response function4}
\end{table} \\
where
\begin{align}
X_+&=\frac{\beta_{5^+}q_{5^+}+\beta_1q_1-\beta_2q_2}{2\beta_{5^+}}~~~~
Y_+=\frac{\beta_{5^+}q_{5^+}-\beta_1q_1+\beta_2q_2}{2\beta_{5^+}}
\end{align}
with $\beta_a=\frac{\tr{\Pi^+_{\{a\}}}}{2}$ and $q_a=1-\frac{2}{\kappa_+}\mu_a$ for $a\in\mc{O_+}$, and
\begin{align}
    X_-&=\frac{\gamma_{5^-}q_{5^-}+\gamma_3q_3-\gamma_4q_4}{2\gamma_{5^-}}~~~~
Y_-=\frac{\gamma_{5^-}q_{5^-}-\gamma_3q_3+\gamma_4q_4}{2\gamma_{5^-}}
\end{align}
with $\gamma_a=\frac{\tr{\Pi^-_{\{a\}}}}{2}$  and $q_a=1-\frac{2}{\kappa_-}\mu_a$ for $a\in\mc{O_-}$. \par 

\subsection{Disputing compatible model using Farka's lemma}
In this paper, we demonstrated the equivalence of POVMs and PVMs in the context of quantum steering. We achieve this by establishing that a compatible model exists for all noisy POVMs whenever such a model exists for all noisy PVMs at the same noise threshold. However, this elegant equivalence disappears when considering scenarios that involve only finite-shared randomness and specific parent POVMs (or specific local hidden states).\par 

In a parallel paper\cite{zhang2023}, we conduct an extensive analysis on the feasibility of a compatible model for noisy PVMs with finite shared randomness. Notably, one intriguing observation we make there is that the optimal parent POVM for simulating the entire set of noisy PVMs may not necessarily be central-symmetric, even if it does exist. The realization that asymmetry holds value in simulating all PVMs provides insight into the potential divergence between POVMs and PVMs within this restricted compatible model.\par 

In the following sections, we provide a rigorous proof for the aforementioned assertion using Farka’s lemma. As we discussed in the main text, We want to show that a solution of $A\cdot \boldsymbol{x}=\boldsymbol{b}$ exists with coefficients $x_{a|i}\geq 0$ that also satisfy the condition $\sum_ax_{a|i}=1$ for all $i=1,\cdots,n$ and $a=1,\cdots,m$, where

\setcounter{MaxMatrixCols}{20}
\begin{align}
    A&=\begin{pmatrix} \vec{\pi}_1&\cdots&\vec{\pi}_n&0&\cdots&0&0&0&0\\
& \ddots &&& \ddots &&& \ddots &\\ 0&0&0&0&\cdots&0&\vec{\pi}_1&\cdots&\vec{\pi}_n\\
    1&\cdots  &0&1&\cdots    &0&1&\cdots  &0\\
    &\ddots  &&&\ddots  &&&\ddots  &\\
    0&\cdots  &1&0&\cdots  &1& 0&\cdots  &1\end{pmatrix},&    
    \boldsymbol{x}&=\begin{pmatrix}x_{1|1}\\ \vdots \\x_{1|n} \\x_{2|n}\\ \vdots\\ x_{m|n}\end{pmatrix},& \boldsymbol{b}&=\begin{pmatrix}\vec{m}_1\\\vdots \\\vec{m}_m\\1\\ \vdots\\ 1\end{pmatrix}.
\end{align}
\setcounter{lemma}{0}
\begin{lemma}[Farkas' Lemma]
\begin{equation}
  \nexists \boldsymbol{x}\ge \boldsymbol{0} \text{ s.t. } A\cdot \boldsymbol{x}=\boldsymbol{b} \Longleftrightarrow \exists y \text{ s.t. } A^T\cdot \boldsymbol{y}\ge \boldsymbol{0} \text{ and }\boldsymbol{b}^T\cdot \boldsymbol{y}<0
\end{equation}
where $\boldsymbol{y}\in\mbb{R}^{4m+n}$ as $y=(\vec{y}_1,\cdots,\vec{y}_m,z_1,\cdots,z_n)^T$. 
\end{lemma}
Therefore, to show the non-existence of a positive, normalized response function $x_{a|i}$ for simulating $\{\wt{M}_a\}$ with $\{\Pi_i\}_{i=1}^n$, it suffices to find a vector $\boldsymbol{y}$ such  that $A^T\cdot \boldsymbol{y}\ge \boldsymbol{0} \text{ and }\boldsymbol{b}^T\cdot \boldsymbol{y}<0$ holds at the same time.

\setcounter{proposition}{0}
\begin{proposition}
There exists a five-effect parent POVM $\{\Pi_i\}_{i=1}^5$ that can simulate all children PVMs with radius $r\le 0.3714$, whereas a compatible model fails to exist for some three-outcome children POVMs with radius $r> 0.3220$
\end{proposition}
\begin{proof}
The 5-effect POVM $\{\Pi_i\}$ to be considered is of the form of:
\begin{align}
\vec{\pi}_1&=0.242(1,\hat{n}_1)^T~~~~\hat{n}_1=(1,0,0)^T \notag\\
\vec{\pi}_2&=0.098(1,\hat{n}_2)^T~~~~\hat{n}_2=(-1, 0, 0)^T \notag\\
\vec{\pi}_3&=0.220(1,\hat{n}_3)^T~~~~\hat{n}_3=(-\frac{12}{55}, \sqrt{1-(\frac{12}{55})^2},0)^T \notag\\
\vec{\pi}_4&=0.220(1,\hat{n}_4)^T~~~~\hat{n}_4=(-\frac{12}{55}, -\frac{1}{2}\sqrt{1-(\frac{12}{55})^2},\frac{\sqrt{3}}{2}\sqrt{1-(\frac{12}{55})^2} )^T \notag\\
\vec{\pi}_5&=0.220(1,\hat{n}_5)^T~~~~\hat{n}_5=(-\frac{12}{55}, -\frac{1}{2}\sqrt{1-(\frac{12}{55})^2},-\frac{\sqrt{3}}{2}\sqrt{1-(\frac{12}{55})^2} )^T
\label{eq:5-parent}
\end{align}
First, we used the criteria for computing the compatible radius $R(\{\Pi_i\})$ in \cite{zhang2023} to obtain a close form expression: 
\begin{equation}
R(\{\Pi_i\})=0.34+0.144\frac{12}{55}\approx 0.3714.
\end{equation}
Therefore, any PVMs with visibility $r\le 0.3714$ can be simulated by the corresponding five-effect POVM. \par 
Now, we give examples showing the infeasibility of simulating three-outcome POVMs $\{M_a^r\}$ with this five-effect parent POVM for an even smaller threshold $r<R(\{\Pi_i\})$.  The specific three-outcome POVM $\{M_a^r\}$ we consider is given by the four-vectors
\begin{align}
\vec{m}_1&=\frac{1}{3}(1,r\hat{m}_1)^T~~~~\hat{m}_1=(0,-1,0)^T \notag \\
\vec{m}_2&=\frac{1}{3}(1,r\hat{m}_2)^T~~~~\hat{m}_2=(0,\frac{1}{2}, \frac{\sqrt{3}}{2})^T \notag\\
\vec{m}_3&=\frac{1}{3}(1,r\hat{m}_3)^T~~~~\hat{m}_3=(0.\frac{1}{2}, -\frac{\sqrt{3}}{2})^T.
\label{eq:counter-3-outcome}
\end{align}
By considering vector $\boldsymbol{y}\in\mbb{R}^{4\times 3+5}$ as $y=(\vec{y}_1,\vec{y}_2,\vec{y}_3,z_1,z_2,z_3,z_4,z_5)^T$, where:
\begin{align}
\vec{y}_1&=(0,-\hat{m}_1)^T,~~~~\vec{y}_2=(0,-\hat{m}_2)^T~~~~\vec{y}_3=(0,-\hat{m}_3)^T\notag \\
z_1&=z_2=z_3=0.110\sqrt{1-(\frac{12}{55})^2},~~~~z_4=z_5=0,
\end{align}
and $A^T\cdot \boldsymbol{y}\ge 0$, one can easily verify that with $ \boldsymbol{b}=(\vec{m}_1,\vec{m}_2,\vec{m}_3,1,1,1,1,1)^T$:
\begin{equation}
   \boldsymbol{b}^T\cdot \boldsymbol{y}=-r+0.330\sqrt{1-(\frac{12}{55})^2}<0\Longrightarrow r>0.330\sqrt{1-(\frac{12}{55})^2}\approx 0.3220
\end{equation}
Therefore, from Farka's lemma, for $r>0.3220$, the three-outcome POVM defined in Eq.~\ref{eq:counter-3-outcome} can never be simulated by the five-effect POVM $\{\Pi_i\}$, whereas all PVMs can be simulated by it with $r<0.3714$.

\end{proof}
\begin{conjecture}
The $5$-outcome POVM in Eq.~\ref{eq:5-parent} is near optimal in terms of simulating PVMs, i.e, close to the optimal $\{\Pi_i\}$ that has the largest compatible radius $R(\{\Pi_i\}) \approx 0.3718$ \cite{zhang2023}.  Therefore, we conjecture that POVMs and PVMs are inequivalent given fixed finite shared randomness (without specifying the local hidden state).
\end{conjecture}
\end{document}